\def\be{\begin{equation}}
\def\ee{\end{equation}}
\def\ba{\begin{array}}
\def\ea{\end{array}}

\documentclass[pra,twocolumn,amsmath]{revtex4}
\usepackage{amsfonts}
\usepackage{mathrsfs}
\usepackage{soul, color, xcolor}
\usepackage{graphicx}
\usepackage{amssymb}
\usepackage{pifont}
\usepackage{epsfig,subfigure,dsfont,amsthm,amsbsy,mathrsfs,amscd}
\usepackage{graphics,graphicx}
\usepackage{epstopdf}
\usepackage{hyperref}
\def\qed{\leavevmode\unskip\penalty9999 \hbox{}\nobreak\hfill
     \quad\hbox{\leavevmode  \hbox to.77778em{%
               \hfil\vrule   \vbox to.675em%
               {\hrule width.6em\vfil\hrule}\vrule\hfil}}
     \par\vskip3pt}

\newtheorem{theorem}{Theorem}

\newtheorem{lemma}{Lemma}

\newtheorem{example}{Example}
\newcommand{\norm}[1]{\lVert#1\rVert}
\input amssym.def

\begin{document}
\title{\large\bf{Separability criteria based on realignment}}

\author{Yu Lu}
\email{2240501015@cnu.edu.cn}
\affiliation{School of Mathematical Sciences, Capital Normal University, Beijing 100048, China}
\author{Zhong-Xi Shen}
\email{18738951378@163.com}
\affiliation{School of Mathematical Sciences, Capital Normal University, Beijing 100048, China}
\author{Shao-Ming Fei}
\email{feishm@cnu.edu.cn}
\affiliation{School of Mathematical Sciences, Capital Normal University, Beijing 100048, China}
\author{Zhi-Xi Wang}
\email{wangzhx@cnu.edu.cn}
\affiliation{School of Mathematical Sciences, Capital Normal University, Beijing 100048, China}

\begin{abstract}
The detection of entanglement in a bipartite state is a crucial issue in quantum information science. Based on realignment of density matrices and the vectorization of the reduced density matrices, we introduce a new set of separability criteria. The proposed separability criteria can detect more entanglement than the previous separability criteria. Moreover, we provide new criteria for detecting the genuine tripartite entanglement and lower bounds for the concurrence and convex-roof extended negativity. The advantages of results are demonstrated through detailed examples.
\end{abstract}

\maketitle

\section{Introduction}

Quantum entanglement \cite{Mintert1679022004,Chen0405042005,Breuer0805012006,Vicente0523202007,Zhang0123342007} is central to the fields of quantum information processing and quantum computation \cite{Nielsen2000}. A primary issue in the field of quantum entanglement is to ascertain whether a quantum state is entangled or not. A bipartite state $\rho\in \mathcal{H}_{M}\otimes \mathcal{H}_{N}$ is said to be separable if it can be represented as a convex sum of tensor products of the states of subsystems,
\begin{equation}\label{sep}
\rho=\sum_{i}p_i\rho^{i}_{M}\otimes \rho^{i}_{N},
\end{equation}
where $p_i\geq 0$ and $\sum_ip_i=1$. Otherwise $\rho$ is said to be entangled.

Consequently, significant efforts have been focused on addressing the so-called separability problem. The most well-known one is the positive partial transpose (PPT) criterion \cite{Peres14131996,Horodecki2231996}, for high-dimensional states, the PPT criterion is only a necessary one. In Refs.\cite{lupo2008bipartite,li2011note}, The authors established separability criteria by analyzing the symmetric functions of the singular values of the realigned matrices, demonstrating their superiority over the existing realignment criterion. In Refs. \cite{Zhang0603012008,shen2015separability,zhang2017realignment,li2017detection},improved realignment criterions have been presented, capable of detecting quantum entanglement. In Ref.\cite{shi2023family}, the authors proposed a series of  separable criteria for bipartite states, which reduce to the improved realignment criterion for particular cases. Later, in Ref.\cite{qi2024detection}, based on the realigned matrix given in Ref.\cite{shi2023family}, a criterion for the separability of multipartite states has been introduced, which detects the genuine multipartite entangled (GME) states. In Ref.\cite{sun2024separability}, the authors improved results and gived the better realignment criterions than Shi $et$ $al.$ gived in Refs.\cite{shi2023family,qi2024detection}.

A key problem in the study of quantum entanglement is the challenge of quantifying entanglement. Various entanglement measures have been presented in recent years \cite{Horodecki2231996,guhne2009entanglement,lee2003convex,huber2013entropy,chen2016lower}. Concurrence and convex-roof extended negativity (CREN) are recognized as two prominent entanglement measures. In Refs. \cite{Chen0405042005,lee2003convex,Vicente0523202007,ma2011measure,li2017measure}, the authors presented improved lower bounds of concurrence and CREN.


In this paper, we introduce a set of separability criteria for bipartite systems, from which we derive tighter lower bounds of concurrence and CREN in Section \ref{S:2}. In Section \ref{S:3}, we generalize our separability criteria to multipartite systems. The criteria detects genuine multipartite entanglement as well as the multipartite full separability.  We have also derived a tighter lower bound for the GME concurrence. The advantages of results are demonstrated through detailed examples. We summarize and conclude in Section \ref{S:4}.


\section{Detection and measures of entanglement for bipartite states}\label{S:2}
\subsection{Separability criteria for bipartite systems}
Let $\mathbb{C}^{m \times n}$ be the set of all $m\times n$ matrices over complex field $\mathbb{C}$, and $\mathbb{R}$ be the real number field. For a matrix $A=\left[a_{i j}\right]\in \mathbb{C}^{m \times n}$, the vectorization of matrix $A$ is defined as
$\operatorname{Vec}(A)=\left(a_{11}, \!\cdots\!, a_{m 1}, a_{12}, \!\cdots\!, a_{m 2}, \!\cdots\!, a_{1 n}, \!\cdots\!, a_{m n}\right)^{T}$, where $T$ stands for the transpose.

Let $Z$ be an $m \times m$ block matrix with sub-blocks $Z_{i, j} \in \mathbb{C}^{n \times n}$, $i, j=1, \ldots, m$. The realigned matrix $\mathcal{R}(Z)$ of $Z$ is defined by
\begin{eqnarray}\notag
	\mathcal{R}(Z)=\left(\begin{array}{c}
		\operatorname{Vec}\left(Z_{1,1}\right)^{T} \\
		\vdots \\
		\operatorname{Vec}\left(Z_{m, 1}\right)^{T} \\
		\vdots \\
		\operatorname{Vec}\left(Z_{1, m}\right)^{T} \\
		\vdots \\
		\operatorname{Vec}\left(Z_{m, m}\right)^{T}
	\end{array}\right).
\end{eqnarray}

The realignment criterion \cite{chen2002matrix} says that any separable state $\rho$ in $\mathbb{C}^{d_{A}} \otimes \mathbb{C}^{d_{B}}$ satisfies
$\|\mathcal{R}(\rho)\|_\mathrm{T r} \leqslant 1$, where $\|A\|_\mathrm{T r}=\operatorname{T r}\left(\sqrt{A^{\dagger} A}\right)$ is the trace norm of $A$.

Then based on the realignment of $\rho_{AB}-\rho_A\otimes\rho_B$, Zhang $et$ $al.$ showed that for any separable state $\rho_{AB}$, the following inequality is valid,
\begin{align}
\|\mathcal{R}(\rho_{AB}-\rho_A\otimes\rho_B)\|_1\le \sqrt{1-\operatorname{T r}\rho_A^2}\sqrt{1-\operatorname{T r}\rho_B^2},
\end{align}
it is stronger than the CCNR criterion \cite{zhang2008entanglement}.\par
By using some parameters and the reduced density matrices of a bipartite state $\rho_{AB},$ Shen $et$ $al.$ \cite{shen2015separability} constructed
\begin{align}
\mathcal{N}_{\beta,l}^{G}(\rho)=\begin{pmatrix}
G&\beta \omega_l(\rho_B)^T\\
\beta\omega_l(\rho_A)&\mathcal{R}(\rho)
\end{pmatrix},
\end{align}
here $G-\beta^2 E_{l\times l}$ is positive semidefinte, $\beta\in \mathbb{R}$, and $\omega_l(X)$ means
\begin{align}
\omega_l(X)=(\underbrace{\mathrm{Vec}(X),\cdots, \mathrm{Vec}(X)}_{\text{$l$ columns}}).
\end{align}
There they showed that when $G-\beta^2 E_{l\times l}\ge 0$, then a separable state $\rho$ satisfies
\begin{align}
\norm{\mathcal{N}_{\beta,l}^{G}(\rho)}_1\le 1+\operatorname{T r}(G).\label{s}
\end{align}
\indent Shi $et$ $al.$ \cite{shi2023family} constructed $\mathcal{M}_{\alpha,\beta}(\rho_{AB})$ on a bipartite state $\rho_{AB}$ as
\begin{align}\label{t1}
\mathcal{M}_{\alpha,\beta}(\rho_{AB})=\begin{pmatrix}
	\alpha\beta& \alpha \mathrm{Vec}(\rho_B)^T\\
	\beta \mathrm{Vec}(\rho_A)& \mathcal{R}(\rho_{AB})
\end{pmatrix},
\end{align}
here $\alpha,\beta\in \mathbb{R}$, $\rho_A$ and $\rho_B$ are redeced density matrices of the $A$ and $B$ system, respectively.
They showed that $\rho_{AB}$ is a separable state, when\begin{align*}
 	\norm{\mathcal{M}_{\alpha,\beta}(\rho_{AB})}_{\mathrm{T r}}\le \sqrt{(\alpha^2+1)(\beta^2+1)}
 	\end{align*}

\indent Sun $et$ $al.$ \cite{sun2024separability} constructed $\mathcal{M}_{\alpha,\beta}^{l}(\rho_{AB})$ on a bipartite state $\rho_{AB}$ as \begin{eqnarray}\label{t2}
	\mathcal{M}_{\alpha, \beta}^{l}(\rho)=\left(\begin{array}{cc}
		\alpha \beta E_{l \times l} & \alpha \omega_{l}\left(\operatorname{T r}_{A}(\rho)\right)^{T} \\
		\beta \omega_{l}\left(\operatorname{T r}_{B}(\rho)\right) & \mathcal{R}(\rho)
	\end{array}\right),
\end{eqnarray}
where $\alpha$ and $\beta$ are arbitrary real numbers, $l$ is a natural number, $E_{l \times l}$ is the matrix with all $l \times l$ elements being $1$, $\operatorname{T r}_{A}$ is the partial trace over the subsystem $A$.
They showed that $\rho\in\mathbb{C}^{d_{A}} \otimes \mathbb{C}^{d_{B}}$ is separable, when
\begin{eqnarray}\notag
\left\|\mathcal{M}_{\alpha, \beta}^{l}\left(\rho\right)\right\|_{\mathrm{T r}} \leq \sqrt{\left(l \alpha^{2}+1\right)\left(l \beta^{2}+1\right)},
\end{eqnarray}

In this manuscript, we denote $\mathcal{Q}_{\mu,\nu}(\rho_{AB})$ on a bipartite state $\rho_{AB}$ as
\begin{eqnarray}\label{eq:m1}
	\mathcal{Q}_{\mu, \nu}(\rho)=\left(\begin{array}{cc}
		\mu\nu^T & \mu \operatorname{Vec}(\rho^i_B)^T \\
		 \operatorname{Vec}(\rho^i_A)\nu^T& \mathcal{R}(\rho)
	\end{array}\right),
\end{eqnarray}
where $\mu=(u_1,...,u_{n})^T$ and $\nu=(v_1,...,v_{m})^T$, where
$u_i$ $(i=1,...,n)$ and $v_j$ $(j=1,...,m)$ are given real numbers, $m$ and $n$ are positive integers.
If $m=n=1,$ let $\alpha = \mu, \beta = \nu$ the equation (\ref{eq:m1}) has degenerated into the matrix (\ref{t1}) constructed in reference \cite{shi2023family}.
If $m=n$ and $u_1=u_2=\cdots = u_{n} = \alpha,v_1=v_2=\cdots = v_{m} = \beta$, the equation (\ref{eq:m1}) has degenerated into the matrix (\ref{t2}) constructed in reference \cite{sun2024separability}

Concerning $\mathcal{Q}_{\mu, \nu}(\rho)$ we have the following lemma, see proof in Appendix A.

\begin{lemma} \label{P:1} $\left(1\right)$ For any $\rho_{i}\in\mathbb{C}^{d_{A}}\!\otimes \mathbb{C}^{d_{B}}$ and $k_{i}\in \mathbb{R} $, $i=1,2, \ldots, n$, such that $\sum\limits_{i=1}^{n} k_{i}=1$, we have
$$
\mathcal{Q}_{\mu, \nu}\left(\sum_{i=1}^{n} k_{i} \rho_{i}\right)=\sum_{i=1}^{n} k_{i} \mathcal{Q}_{\mu, \nu}\left(\rho_{i}\right).
$$
$\left(2\right)$ Let $U$ and $V$ be unitary matrices on subsystems of A and B, respectively. For any $\rho\in\mathbb{C}^{d_{A}} \otimes \mathbb{C}^{d_{B}}$, we have $\left\| \mathcal{Q}_{\mu, \nu}\left((U \otimes V) \rho(U \otimes V)^{\dagger}\right)\left\|_{\mathrm{T r}}=\right\| \mathcal{Q}_{\mu, \nu}\left(\rho\right) \|_{\mathrm{T r}}\right.$.
\end{lemma}

By Lemma \ref{P:1} we have the following separability criterion.

\begin{theorem}\label{th:1} If a state $\rho\in\mathbb{C}^{d_{A}} \otimes \mathbb{C}^{d_{B}}$ is separable, then
\begin{eqnarray}\notag
\left\|\mathcal{Q}_{\mu, \nu}\left(\rho\right)\right\|_{\mathrm{T r}} \leq \sqrt{\left( |\mu|^{2}+1\right)\left(|\nu|^{2}+1\right)},
\end{eqnarray}
where $\mathcal{Q}_{\mu,\nu}\left(\rho\right)$ is defined in (\ref{eq:m1}).
\end{theorem}

\begin{proof}
Since $\rho $ is separable, it can be written as a convex combination of pure states,
$\rho=\sum_{i} p_{i} \rho_{A}^{i} \otimes \rho_{B}^{i}$, where $p_{i} \in[0,1]$ with $\sum\limits_{i} p_{i}=1$, $\rho_{A}^{i}$ and $\rho_{B}^{i}$ are pure states of the subsystems $A$ and $B$, respectively. From Lemma \ref{P:1} we have
\begin{eqnarray}\label{eq:4}\notag
&& \left\|\mathcal{Q}_{\mu, \nu}\left(\rho\right)\right\|_{\mathrm{T r}}\\\notag&=&\left\|\sum_{i} p_{i} \mathcal{Q}_{\mu, \nu}\left(\rho_{A}^{i} \otimes \rho_{B}^{i}\right)\right\|_{\mathrm{T r}} \\&\leq& \sum_{i} p_{i}\left\|\mathcal{Q}_{\mu, \nu}\left(\rho_{A}^{i} \otimes \rho_{B}^{i}\right)\right\|_{\mathrm{T r}}.
\end{eqnarray}

Since for any $A\in \mathbb{C}^{m \times n}$ and $B \in \mathbb{C}^{p \times q}$,
\begin{eqnarray}\label{eq:r4}
\mathcal{R}(A \otimes B)=\operatorname{Vec}(A)\operatorname{Vec}(B)^{T},
\end{eqnarray}
we obtain from the definition of $\mathcal{Q}_{\mu,\nu}\left(\rho\right)$,
	\begin{eqnarray}\label{eq:6}\notag
			& & \left\|\mathcal{Q}_{\mu, \nu}\left(\rho_A^i \otimes \rho_B^i\right)\right\|_{\mathrm{T r}} \\\notag
			&=&\left\|\left(\begin{array}{cc}
				\mu \nu^T& \mu \operatorname{Vec}\!\left(\!\operatorname{T r}_A\left(\rho_A^i \otimes \rho_B^i \!\right)\!\right)^T \\
				\operatorname{Vec}\left(\!\operatorname{T r}_B\left(\rho_A^i \otimes \rho_B^i\!\right)\right)\nu^T & \mathcal{R}\left(\rho_A^i \otimes \rho_B^i\!\right)
			\end{array}\right)\right\|_{\mathrm{T r}} \\\notag
			& =&\left\|\left(\begin{array}{cc}
				\mu \nu^T& \mu \operatorname{Vec}\left(\rho_B^i\right)^T \\
				\operatorname{Vec}\left(\rho_A^i\right)\nu^T & \operatorname{Vec}\left(\rho_A^i\right) \operatorname{Vec}\left(\rho_B^i\right)^T
			\end{array}\right)\right\|_{\mathrm{T r}} \\
			& =&\|\left(\begin{array}{c}
				\mu\\
				\operatorname{Vec}\left(\rho_A^i\right)
			\end{array}\right)\left(\begin{array}{cc}
				\nu^T \quad\left.\operatorname{Vec}\left(\rho_B^i\right)^T\right)
			\end{array} \|_{\mathrm{T r}}.\right.
	\end{eqnarray}
As $\rho_{A}^{i}$ and $\rho_{B}^{i}$ are pure states, i.e., $\operatorname{T r}\left(\rho_A^i\right)^2=\operatorname{T r}\left(\rho_B^i\right)^2=1$, we have
	\begin{eqnarray}\label{eq:27}\notag
			& &\left\|\left(\begin{array}{c}
				\mu \\
				\operatorname{Vec}\left(\rho_A^i\right)
			\end{array}\right)\left(\nu^T\quad \operatorname{Vec}\left(\rho_B^i\right)^T\right)\right\|_{\mathrm{T r}} \\\notag
			& =&\sqrt{|\nu|^{2}+1}\operatorname{T r}\left(\left(\begin{array}{c}
				\mu  \\
				\operatorname{Vec}\left(\rho_A^i\right)
			\end{array}\right)\left(\mu \quad \operatorname{Vec}\left(\rho_A^i\right)^T\right)\right)^{\frac{1}{2}} \\
			& =&\sqrt{(|\nu|^{2}+1)} \sqrt{(|\mu|^{2}+1)}.
	\end{eqnarray}
Combining (\ref{eq:4}), (\ref{eq:6}) and (\ref{eq:27}) we get
	\begin{eqnarray}\notag
			&&\left\|\mathcal{Q}_{\mu, \nu}\left(\rho\right)\right\|_{\mathrm{T r}}\\\notag& \leq& \sum_{i} p_{i}\left\|\mathcal{Q}_{\mu, \nu}\left(\rho_{A}^{i} \otimes \rho_{B}^{i}\right)\right\|_{\mathrm{T r}}
			\\\notag&=&\sqrt{(|\nu|^{2} +1} )\sqrt{(|\mu|^{2}+1)},
	\end{eqnarray}
which completes the proof.
\end{proof}

We provide two examples to illustrate the Theorem \ref{th:1}.

\begin{example}\label{ex:1}
Consider the state $\rho_{x}=x\left|\xi \right\rangle\left\langle\xi\right|+ \left(1-x\right)\rho_{d}$, where $\rho_{d}$ is the $2 \times 4$ bound entangled state,
$$
	\rho_{d}=\frac{1}{1+7d}
	\begin{pmatrix}
		0 & d & 0 & 0 & 0 & 0 & d & 0 \\
		0 & 0 & d & 0 & 0 & 0 & 0 & d \\
		0 & 0 & 0 & d & 0 & 0 & 0 & 0 \\
		0 & 0 & 0 & 0 & \frac{1+d}{2} & 0 & 0 & \frac{\sqrt{1-d^2}}{2} \\
		d & 0 & 0 & 0 & 0 & d & 0 & 0\\
		0 & d & 0 & 0 & 0 & 0 & d & 0\\
		0 & 0 & d & 0 & \frac{\sqrt{1-d^2}}{2} & 0 & 0 & \frac{1+d}{2} \\
	\end{pmatrix},
$$
where $0 \textless d\textless 1$, and $\left|\xi \right\rangle=\frac{1}{\sqrt{2}}(|00\rangle+|11\rangle)$.
\end{example}

Set $d=0.9$ and take
\begin{align}
\mu =( 11.9967, 12.9195, 11.6808, 12.1705, 11.4476)^T,\\
\nu = (12.5025, 11.5102, 12.0119, 12.3982, 12.7818)^T,
 \end{align}
 by direct calculation we get from Theorem \ref{th:1} that $\rho_{x}$ is entangled for $0.232959 \leq x \leq 1$. Let
 \begin{align}
  \mu =(11.66, 11.66, 11.66, 11.66, 11.66)^T,\\
  \nu = (11.75, 11.75, 11.75, 11.75, 11.75)^T,
 \end{align}
 we will get the results of $\rho_{x}$ is entangled for $0.233889 \leq x \leq 1$ from Theorem $1$ in Ref.\cite{sun2024separability}. Set $\mu = 11.66$ and $\nu = 11.75$, we will get the results of $\rho_{x}$ is entangled for $0.233931 \leq x \leq 1$ from Theorem $1$ in Ref.\cite{shi2023family}. Obviously our Theorem \ref{th:1} detects the entanglement of the state $\rho_{x}$ better than them.

\begin{example}\label{ex:2}
Consider the mixture of the bound entangled state proposed by Horodecki \cite{horodecki1997separability},
	$$
	\rho_{t}=\frac{1}{1+8t}
	\begin{pmatrix}
		t & 0 & 0 & 0 & t & 0 & 0 & 0 & t\\
		0 & t & 0 & 0 & 0 & 0 & 0 & 0 & 0\\
		0 & 0 & t & 0 & 0 & 0 & 0 & 0 & 0\\
		0 & 0 & 0 & t & 0 & 0 & 0 & 0 & 0 \\
		t & 0 & 0 & 0 & t & 0 & 0 & 0 &  t\\
		0 & 0 & 0 & 0 & 0 & t & 0 & 0 & 0\\
		0 & 0 & 0 & 0 & 0 & 0 & \frac{1+t}{2} & 0 & \frac{\sqrt{1-t^2}}{2}\\
0 & 0 & 0 & 0 & 0 & 0 & 0 & t & 0\\
		t & 0 & 0 & 0 & t & 0 & \frac{\sqrt{1-t^2}}{2} & 0 & \frac{1+t}{2}\\
	\end{pmatrix}
	$$
and the $9\times9$ identity matrix $I_9$,
$$
\rho_{p}=\frac{1-p}{9} I_{9}+p \rho_{t}.
$$
\end{example}

We take $\mu=(\frac{37}{16},2 ,2,2,2,2,2,2,2,2)^T$ and $\nu=(\frac{47}{20},2 ,2,2,2,2,2,2,2,2)^T$. We compare among the results from our Theorem \ref{th:1}, realignment criterion from Theorem $1$ in Ref.\cite{shi2023family} and Ref.\cite{sun2024separability}  for different values of $t$, see Table \ref{tab:1}. Table \ref{tab:1} shows that our Theorem \ref{th:1} detects better the entanglement of the state $\rho_{p}$ than the criteria from \cite{shi2023family} and \cite{sun2024separability}
\setlength{\tabcolsep}{15pt}
\begin{center}
	\begin{table*}[ht]
\caption{Entanglement of the state $\rho_{p}$ in Example \ref{ex:2} for different values of $t$}
		\label{tab:1}
		\begin{tabular}{cccc}
			\hline\noalign{\smallskip}
		$t$ & Theorem $1$ in \cite{shi2023family} &realignment in \cite{sun2024separability}  & Our Theorem \ref{th:1}  \\
		\noalign{\smallskip}\hline\noalign{\smallskip}
		0.2 & $0.9943 \leq p \leq 1$ & $0.9942  \leq p \leq 1$ & $0.9940  \leq p \leq 1$ \\
		
		0.4 & $0.9948  \leq p \leq 1$ & $0.9947  \leq p \leq 1$ & $0.9946  \leq p \leq 1$ \\
		
		0.6 & $0.9964  \leq p \leq 1$ & $0.9963  \leq p \leq 1$ & $0.99625 \leq p \leq 1$ \\
		
		0.8 & $0.9982  \leq p \leq 1$ & $0.99815  \leq p \leq 1$ & $0.99813  \leq p \leq 1$ \\
		
		0.9 & $0.9991  \leq p \leq 1$ & $0.99908  \leq p \leq 1$ & $0.99907  \leq p \leq 1$ \\
		\noalign{\smallskip}\hline
		\end{tabular}
	\end{table*}
\end{center}

\subsection{Lower bounds of concurrence and CREN for bipartite states}
The concurrence of a pure state $|\varphi\rangle\in\mathbb{C}^{d_{A}} \otimes \mathbb{C}^{d_{B}}$ is defined by \cite{rungta2001universal}
\begin{eqnarray}\label{C:1}
	C\left(|\varphi\rangle\right)=\sqrt{2\left(1-\operatorname{T r} \left(\rho_{A}^{2}\right)\right)},
\end{eqnarray}
where $\rho_{A}=\operatorname{T r}_{B}\left(|\varphi\rangle\langle\varphi|\right)$. The concurrence of a mixed state $\rho$ is defined as
\begin{eqnarray}\label{C:2}
	C\left(\rho\right)=\min _{\left\{p_{i}, \left|\varphi_{i}\right\rangle\right\}} \sum_{i} p_{i} C\left(\left|\varphi_{i}\right\rangle\right),
\end{eqnarray}
where the minimum is taken over all possible ensemble decompositions of
$\rho=\sum\limits_{i}p_{i}\left|\varphi_{i}\right\rangle\left\langle\varphi_{i}\right|$, $p_{i} \geqslant 0$ with $\quad\sum\limits_{i}p_{i}=1$.

The CREN  of a pure state $|\varphi\rangle\in \mathbb{C}^{d_{A}} \otimes \mathbb{C}^{d_{B}}$ is defined by \cite{lee2003convex}
\begin{eqnarray}\notag
	 \mathcal{N}\left(|\varphi\rangle\right)
=\frac{\left\|\left(|\varphi\rangle\langle\varphi|\right)^{T_{B}}\right\|_\mathrm{T r}-1}{d-1},
\end{eqnarray}
where $d=\min \left(d_A, d_B\right)$, $\left(|\varphi\rangle\langle\varphi|\right)^{T_{B}}$ denotes the partial transpose of $|\varphi\rangle\langle\varphi|$. For a mixed state $\rho$, its CREN is defined via convex roof extension,
\begin{eqnarray}\label{N:2}
	\mathcal{N}\left(\rho\right)=\min _{\left\{p_{i},\left|\varphi_{i}\right\rangle\right\}} \sum_{i} p_{i} \mathcal{N}\left(\left|\varphi_{i}\right\rangle\right),
\end{eqnarray}
where the minimum is taken over all possible pure state decompositions of $\rho$.

To derive the lower bounds of concurrence and CREN for arbitrary density matrices, we first present the following lemma, see proof in Appendix B.

\begin{lemma}\label{le:2} Let $|\varphi\rangle$ be a pure bipartite state in systems A and B, with Schmidt decomposition $|\varphi\rangle=\sum\limits_{i=0}^{d-1} \sqrt{\lambda_{i}}\left|i_{A} i_{B}\right\rangle$, where $d=\min \left(d_{A}, d_{B}\right)$. Then

\noindent(1)
$\left\|\mathcal{Q}_{\mu, \nu}\left(|\varphi\rangle\langle\varphi|\right)\right\|_\mathrm{T r}
\\\leq \sqrt{\left(|\mu|^{2}+1\right)\left(| \nu|^{2}+1\right)} +2\sum\limits_{0 \leq i<j \leq d-1} \sqrt{\lambda_{i} \lambda_{j}}$,

\noindent(2)	
	$ \left\|\mathcal{Q}_{\mu, \nu}(|\varphi\rangle\langle\varphi|)\right\|_{\mathrm{T r}}\leq \sqrt{\left(|\mu|^{2}+1\right)\left(| \nu|^{2}+1\right)}+(d-1)$.
\end{lemma}

According to the above lemma, we have

\begin{theorem} \label{th:2}
For a bipartite state $\rho\in\mathbb{C}^{d_{A}} \otimes \mathbb{C}^{d_{B}}$, the concurrence satisfies that
\begin{eqnarray}\notag
C \!\left(\rho\right)\geq \frac{\sqrt{2}}{\sqrt{d(d-1)}}\!\!\left(\left\|\mathcal{Q}_{\mu, \nu}\left(\rho\right)\right\|_\mathrm{T r}\!\! - \!\sqrt{\left(|\mu|^{2}+1\!\right)\left(|\nu|^{2}+1\!\right)}\right),
\end{eqnarray}
where $d=\min \left(d_{A}, d_{B}\right)$.
\end{theorem}

\begin{proof}
For any $|\varphi\rangle\in\mathbb{C}^{d_{A}} \otimes \mathbb{C}^{d_{B}}$ with Schmidt form $|\varphi\rangle=\sum\limits_{i=0}^{d-1} \sqrt{\lambda_{i}}\left|i_{A} i_{B}\right\rangle$, one has \cite{chen2005concurrence}
	\begin{eqnarray}\notag
		C^{2}(|\varphi\rangle) \geq \frac{8}{d(d-1)}\left(\sum_{0 \leq i<j \leq d-1} \sqrt{\lambda_{i} \lambda_{j}}\right)^{2}.
	\end{eqnarray}
	From (1) of Lemma \ref{le:2}, we get
	\begin{eqnarray}\notag
			&& \ \ \ \ C(|\varphi\rangle) \\\notag&&\geq \frac{2 \sqrt{2}}{\sqrt{d(d-1)}} \sum_{0 \leq i<j \leq d-1} \sqrt{\lambda_{i} \lambda_{j}} \\\notag
			&& \geq \!\!\frac{\sqrt{2}}{\sqrt{d(d-1)}}\!\left(\!\left\|\mathcal{Q}_{\mu, \nu}\!\left(|\varphi\rangle\langle\varphi|\right)\right\|_\mathrm{T r}\!\!\!\!
 - \!\sqrt{\left(|\mu|^{2}\!+\!1\right)\!\left(| \nu|^{2}\!+\!1\right)}\!\right).
	\end{eqnarray}

Let $\left\{p_{i},\left|\varphi_{i}\right\rangle\right\}$ be the optimal decomposition of $\rho$ such that $C\left(\rho\right)=\sum\limits_{i} p_{i} C\left(\left|\varphi_{i}\right\rangle\right)$. From Lemma \ref{P:1} we have
	\begin{eqnarray}\notag
		\begin{aligned}
			C\left(\rho\right)&=\sum_{i} p_{i} C\left(\left|\varphi_{i}\right\rangle\right) \\
			&\geq \frac{\sqrt{2}}{\sqrt{d(d-1)}} \sum_{i} p_{i}\left(\left\|\mathcal{Q}_{\mu, \nu}\left(\left|\varphi_{i}\right\rangle\left\langle\varphi_{i}\right|\right)\right\|_{\mathrm{T r}}\right.\\&\left.\quad-\sqrt{\left(|\mu|^{2}+1\right)\left(| \nu|^{2}+1\right)}\right) \\
			& \geq \frac{\sqrt{2}}{\sqrt{d(d-1)}}\left(\left\|\mathcal{Q}_{\mu, \nu}\left(\rho\right)\right\|_{\mathrm{T r}}\right.\\&\left.\quad-\sqrt{\left(|\mu|^{2}+1\right)\left(|\mu|^{2}+1\right)}\right).
		\end{aligned}
	\end{eqnarray}
Therefore, Theorem \ref{th:2} holds.
\end{proof}

\begin{theorem}\label{th:3} For the CREN of any state $\rho\in\mathbb{C}^{d_{A}} \otimes \mathbb{C}^{d_{B}}$, we have	
\begin{eqnarray}\notag
		\mathcal{N}\left(\rho\right) \geq \frac{\left\|\mathcal{Q}_{\mu, \nu}\left(\rho\right)\right\|_{\mathrm{T r}}-\sqrt{\left(|\mu|^{2}+1\right)\left(| \nu|^{2}+1\right)}}{d-1},
	\end{eqnarray}
where $d=\min \left(d_{A}, d_{B}\right)$.
\end{theorem}

\begin{proof}
For any pure state $|\varphi\rangle\in\mathbb{C}^{d_{A}} \otimes \mathbb{C}^{d_{B}}$ with Schmidt form $|\varphi\rangle=\sum\limits_{i=0}^{d-1} \sqrt{\lambda_{i}}\left|i_{A} i_{B}\right\rangle$, one has \cite{vidal2002computable}
	\begin{eqnarray}\notag
		\mathcal{N}\left(|\varphi\rangle\right)=\frac{2}{d-1} \sum_{0 \leq i<j \leq d-1} \sqrt{\lambda_{i} \lambda_{j}}.
	\end{eqnarray}
Using (1) of Lemma \ref{le:2}, we get
	\begin{eqnarray}\notag
		&& 2 \sum_{0 \leq i<j \leq d-1} \sqrt{\lambda_{i} \lambda_{j}} \\\notag&\geq&\left\|\mathcal{Q}_{\mu, \nu}\left(|\varphi\rangle\langle\varphi|\right)\right\|_{\mathrm{T r}}-\sqrt{\left(|\mu|^{2}+1\right)\left(|\nu|^{2}+1\right)}.
	\end{eqnarray}
Let $\left\{p_{i},\left|\varphi_{i}\right\rangle\right\}$ be the optimal pure state decomposition of $\rho$ such that $\mathcal{N}(\rho)=\sum\limits_{i} p_{i} \mathcal{N}\left(\left|\varphi_{i}\right\rangle\right)$. Then using Lemma 1 we obtain
	\begin{eqnarray}\notag
	\begin{aligned}
			\mathcal{N}\left(\rho\right)&=\sum_{i} p_{i} \mathcal{N}\left(\left|\varphi_{i}\right\rangle\right) \\
			&\geq \sum_{i} p_{i} \frac{\left\|\mathcal{Q}_{\mu, \nu}\left(\left|\varphi_{i}\right\rangle\left\langle\varphi_{i}\right|\right)\right\|_{\mathrm{T r}}-\sqrt{\left(|\mu|^{2}+1\right)\!\left(|\nu|^{2}+1\right)}}{d-1} \\
			&\geq \frac{\left\|\mathcal{Q}_{\mu, \nu}\left(\rho\right)\right\|_{\mathrm{T r}}-\sqrt{\left(|\mu|^{2}+1\right)\left(| \nu|^{2}+1\right)}}{d-1},
		\end{aligned}
	\end{eqnarray}
which completes the proof.
\end{proof}

\begin{example}
The following $3 \times 3$ $P P T$ entangled state was introduced in Ref. \cite{bennett1999unextendible},
$$
\rho=\frac{1}{4}\left(I_{9}-\sum_{i=0}^{4}
\left|\varphi_{i}\right\rangle\left\langle\varphi_{i}\right|\right),
$$
where
$$
\begin{aligned}
	& \left|\varphi_{0}\right\rangle=\frac{\sqrt{2}|0\rangle(|0\rangle-|1\rangle)}{2}, \\
	& \left|\varphi_{1}\right\rangle=\frac{\sqrt{2}(|0\rangle-|1\rangle)|2\rangle}{2}, \\
	& \left|\varphi_{2}\right\rangle\frac{\sqrt{2}|2\rangle(|1\rangle-|2\rangle)}{2}, \\
	& \left|\varphi_{3}\right\rangle=\frac{\sqrt{2}(|1\rangle-|2\rangle)|0\rangle}{2}, \\
	& \left|\varphi_{4}\right\rangle=\frac{(|0\rangle+|1\rangle+|2\rangle)(|0\rangle+|1\rangle+|2\rangle) }{3}.
\end{aligned}
$$
\end{example}
When choosing $\mu=(1,1)^T, \nu=(1,0)^T$ according to Theorem \ref{th:2}, $C(\rho_{AB})\ge 0.04407.$ By using the theorem 6 in Ref.\cite{shi2023family}, we have $C(\rho)\ge 0.05399$, hence our bound is better than them. \par
Next we consider a state by mixing $\rho$ with the white noise,
Let us consider the mixture of $\rho$ with white noise,
$$
\rho_{t}=\frac{1-t}{9} I_{9}+t \rho.
$$
We take $\mu = (\frac{2227}{347},\frac{4236}{571},\frac{2233}{345})^T, \nu= (\frac{6819}{1093}, \frac{1580}{219}, \frac{2491}{361})^T$, which shows that $\rho_{t}$ is entangled for $0.88221 \leq t \leq\ 1$. Theorem $2$ in Ref. \cite{sun2024separability}, from which $\rho_{t}$ is entangled for  $0.88248 \leq t \leq\ 1$. If $m = n = 1,$ and $\mu = \nu = 1$, $\rho_{t}$ is entangled for $0.88438 \leq t \leq\ 1$, which is the results of Theorem $6$ in Ref.\cite{shi2023family}, hence our bound is better than them.

\section{Detection and measures of multipartite entanglement}\label{S:3}
\subsection{Separability criteria for multipartite states}

We first consider tripartite case. Denote the three bipartitions of a tripartite quantum state $\rho\in\mathbb{C}^{d_{1}} \otimes \mathbb{C}^{d_{2}} \otimes \mathbb{C}^{d_{3}}$ as $1|23$, $2|13$ and $3|12$. If a tripartite state $\rho$ is biseparable \cite{jing2022criteria}, then
\begin{eqnarray}\notag
	\begin{aligned}
	&\rho=\sum_i p_i|\varphi_i\rangle^{1|23}\langle\varphi_i|+\sum_j p_j|\varphi_j\rangle^{2|13}\langle\varphi_j|\\&\quad\quad+\sum_k p_k|\varphi_k\rangle^{3|12}\langle\varphi_k|,
	\end{aligned}
\end{eqnarray}
where $p_i, p_j, p_k \geq 0$, with $\sum\limits_i p_i+\sum\limits_j p_j+\sum\limits_k p_k=1$. Otherwise, $\rho$ is called genuinely tripartite entangled.
We define
\begin{eqnarray}\label{M:1} \mathscr{Q}(\rho)\!\!=\!\frac{1}{3}\left(\left\|\mathscr{Q}_{1|23}
\!\left(\rho\right)\!\right\|_{\mathrm{T r}}\!\!+\!\left\|\mathscr{Q}_{2|13}\!\left(\rho\right)\!\right\|_{\mathrm{T r}}\!\!+\!\left\|\mathscr{Q}_{3|12}\!\left(\rho\right)\!\right\|_{\mathrm{T r}}\!\right),
\end{eqnarray}
where $\mathscr{Q}_{i|jk}$ stands for the matrix (\ref{eq:m1}) under bipartition $i$ and $jk$, $\{i, j, k\}=\{1,2,3\}$.

\begin{theorem}\label{th:6}
If a tripartite state $\rho\in\mathbb{C}^{d_{1}} \otimes \mathbb{C}^{d_{2}} \otimes \mathbb{C}^{d_{3}}$ is biseparable, then
	\begin{eqnarray}\notag
		\mathscr{Q}(\rho) \leq \sqrt{\left(|\mu|^{2}+1\right)\left(| \nu|^{2}+1\right)}+\frac{2(d-1)}{3},
	\end{eqnarray}
where $d_i=d$ $\left(i=1,2,3\right)$.
\end{theorem}

\begin{proof}
Let $|\varphi\rangle\in\mathbb{C}^{d_{1}} \otimes \mathbb{C}^{d_{2}} \otimes \mathbb{C}^{d_{3}}$ be biseparable under bipartition $1|23$, i.e., $|\varphi\rangle=\left|\varphi_{1}\right\rangle \otimes\left|\varphi_{23}\right\rangle$. From Theorem \ref{th:1} and (2) of Lemma \ref{le:2} we have	
\begin{eqnarray}\label{M:2}
		\begin{aligned}
			& \quad\mathscr{Q}(|\varphi\rangle\langle\varphi|)\\
			& =\frac{1}{3}\left(\left\|\mathscr{Q}_{1|23}\left(|\varphi\rangle\langle\varphi|\right)\right\|_{\mathrm{T r}}+\left\|\mathscr{Q}_{2|13}\left(|\varphi\rangle\langle\varphi|\right)\right\|_{\mathrm{T r}}\right.\\&\left.\quad+\left\|\mathscr{Q}_{3|12}\left(|\varphi\rangle\langle\varphi|\right)\right\|_{\mathrm{T r}}\right) \\
			& \leq \sqrt{\left(|\mu|^{2}+1\right)\left(| \nu|^{2}+1\right)}+\frac{2(d-1)}{3},
		\end{aligned}
	\end{eqnarray}
which holds also for biseparable states under bipartitions $2|13$ and $3|12$.
	
For biseparable mixed state, $\rho=\sum_{i} p_{i}\left|\varphi_{i}\right\rangle\left\langle\varphi_{i}\right|$, $p_{i} \geq 0$, $\sum_{i} p_{i}=1$, we have
	\begin{eqnarray}\notag
		\mathscr{Q}(\rho)\leq\sum_{i} p_{i}\mathscr{Q}(|\varphi_{i}\rangle\langle\varphi_{i}|),
	\end{eqnarray}
which gives rise to that
	\begin{eqnarray}\notag
		\begin{aligned}
			\mathscr{Q}(\rho)
			& \leq \sum_{i} p_{i}\left( \sqrt{\left(|\mu|^{2}+1\right)\left(| \nu|^{2}+1\right)}+\frac{2(d-1)}{3}\right) \\
			& =\sqrt{\left(|\mu|^{2}+1\right)\left(| \nu|^{2}+1\right)}+\frac{2(d-1)}{3}.
		\end{aligned}
	\end{eqnarray}
by using from (\ref{M:2}).
\end{proof}

We provide an example to illustrate the Theorem \ref{th:6}.
\begin{example}
Consider the state $\rho_{W}^{q}$ in $\mathbb{C}^3 \otimes \mathbb{C}^3 \otimes \mathbb{C}^3$,
	$$
	\rho_{W}^{q}=\frac{1-q}{27} I_{27}+q\left|\varphi_{W}\right\rangle\left\langle\varphi_{W}\right|,
	$$
	where $\left|\varphi_{W}\right\rangle=\frac{1}{\sqrt{6}}(|001\rangle+|010\rangle+|100\rangle+|112\rangle+|121\rangle+|211\rangle)$, $I_{27}$ is the $27 \times 27$ identity matrix.
\end{example}

In this case $d=3$. When $\mu=(1,2)^T$ and $\nu=(2,1)^T$, Theorem \ref{th:6} detects the genuine tripartite entanglement of $\rho_{W}^{q}$ for $0.805132 \leq q \leq 1$. When $\mu=\nu=(\sqrt{2},\sqrt{2})^T$, Theorem \ref{th:6} reduces to the Theorem $4$ in Ref.\cite{sun2024separability}, which detects the genuine tripartite entanglement of $\rho_{W}^{q}$ for $0.805211 \leq q \leq 1$. When $\mu=\nu= 1$, Theorem \ref{th:6} reduces to the Theorem $2$ in Ref.\cite{qi2024detection}, which detects the genuine tripartite entanglement of $\rho_{W}^{q}$ for $0.805321 \leq q \leq 1$. Obviously, our Theorem \ref{th:6} is more effective in detecting the genuine tripartite entanglement. $\\$

Next we consider the fully separability of general multipartite states. Any multipartite state $\rho\in\mathbb{C}^{d_{1}} \otimes \mathbb{C}^{d_{2}} \otimes \cdots \otimes \mathbb{C}^{d_{n}}$ can be written as
\begin{eqnarray}\notag
	\rho=\sum_{i} Y_{1}^{i} \otimes Y_{2}^{i} \otimes \cdots \otimes Y_{n}^{i},
\end{eqnarray}
where $Y_{j}^{i} \in \mathbb{C}^{d_{j} \times d_{j}}$, $j=1,2, \cdots, n$. We define
\begin{eqnarray}\label{eq:13}\notag
&&\mathcal{Q} \mathcal{R}_{\mu_{q}, \ldots, \mu_{n}}^{l, q}(\rho)\\&=&\sum_{i} Y_{1}^{i} \otimes \!\cdots \!\otimes Y_{q-1}^{i} \otimes\! \mathcal{Q}_{\mu_{q}, \ldots, \mu_{n}}^{l}\left[Y_{q}^{i}, \ldots, Y_{n}^{i}\right],
\end{eqnarray}
where $q=1,\ldots,n-1$, $\mu_{k} = (u_{k_{1}}, u_{k_2}, \ldots, u_{k_m})^T, k = q, q+1, \ldots, n,$ $u_{k_j} (j = 1, 2, \ldots, m)$are given real numbers,
\begin{eqnarray}\notag
	&&\mathcal{Q}_{\mu_{q}, \ldots, \mu_{n}}^{l}\left[Y_{q}^{i}, \ldots, Y_{n}^{i}\right]\\\notag&=&\left(\begin{array}{c}
		\mu_{q} \\
		\operatorname{Vec}\left(Y_{q}^{i}\right)
	\end{array}\right) \bigotimes_{j=0}^{n-q-1}\left(\mu_{n-j}^T \quad \operatorname{Vec}\left(Y_{n-j}^{i}\right)^{T}\right),
\end{eqnarray}
We have the following separability criterion for multipartite states.

\begin{theorem}\label{th:4}
If a multipartite state  $\rho\in\mathbb{C}^{d_{1}} \otimes \mathbb{C}^{d_{2}} \otimes \cdots \otimes \mathbb{C}^{d_{n}}$ is fully separable, then for any $1\leq{q}\leq{n}-1$,
\begin{eqnarray}\notag
\left\|\mathcal{Q} \mathcal{R}_{\mu_{q}, \ldots, \mu_{n}}^{q}(\rho)\right\|_{\mathrm{T r}} \leq \prod_{k=q}^{n} \sqrt{|\mu_{k}|^{2}+1},
	\end{eqnarray}
where $\mathcal{Q} \mathcal{R}_{\mu_{q},  \ldots, \mu_{n}}^{q}(\rho)$ is defined in (\ref{eq:13}).
\end{theorem}

\begin{proof}
Since any fully separable state $\rho$ can be written as
$\rho=\sum_{i} p_{i} \rho_{1}^{i} \otimes \rho_{2}^{i} \otimes \cdots \otimes \rho_{n}^{i}$,
where $p_{i} \in[0,1]$ with $\sum\limits_{i} p_{i}=1$, we have
\begin{eqnarray}\notag
\begin{aligned}
&\quad\left\|\mathcal{Q} \mathcal{R}_{\mu_{q}, \ldots, \mu_{n}}^{ q}(\rho)\right\|_{\mathrm{T r}}\\&=\left\|\sum\limits_{i} p_{i} \rho_{1}^{i} \otimes \cdots \otimes \rho_{q-1}^{i} \otimes \mathcal{Q}_{\mu_{q}, \cdots, \mu_{n}}\left[\rho_{q}^{i}, \ldots, \rho_{n}^{i}\right]\right\|_{\mathrm{T r}}  \\
			&\leq \sum\limits_{i}p_{i}\left\|\rho_{1}^{i} \otimes \cdots \otimes \rho_{q-1}^{i} \otimes \mathcal{Q}_{\mu_{q}, \ldots, \mu_{n}}\left[\rho_{q}^{i}, \ldots, \rho_{n}^{i}\right]\right\|_{\mathrm{T r}}
			\\&=\sum\limits_{i} p_{i}\operatorname{T r}\left(\rho_{1}^{i}\right) \cdots\operatorname{T r}\left(\rho_{q-1}^{i}\right)\prod_{k=q}^{n}\left\|\left(\begin{array}{c}
				\mu_{k} \\
				\operatorname{Vec}\left(\rho_{k}^{i}\right)
			\end{array}\right)\right\|_{\mathrm{T r}}\\
			&=\sum\limits_{i} p_{i}\prod_{k=q}^{n} \sqrt{| \mu_{k}|^{2}+1}\\&=\prod_{k=q}^{n} \sqrt{|\mu_{k}|^{2}+1},
		\end{aligned}
	\end{eqnarray}
which proves Theorem \ref{th:4}.
\end{proof}

\subsection{Lower bounds of GME concurrence for tripartite systems}
The GME concurrence of a pure state $|\varphi\rangle\in\mathbb{C}^{d} \otimes \mathbb{C}^{d} \otimes \mathbb{C}^{d}$ is defined by \cite{ma2011measure},
\begin{eqnarray}\notag
	 C_{G M E}(|\varphi\rangle)=\sqrt{\!\min \left\{1-\operatorname{T r}\left(\rho_{1}^{2}\right), 1-\operatorname{T r}\left(\rho_{2}^{2}\right), 1-\operatorname{T r}\left(\rho_{3}^{2}\right)\right\}},
\end{eqnarray}
where $\rho_{i}$ is the reduced density matrix of subsystem $i$.
The GME concurrence of a mixed state $\rho$ is defined as
\begin{eqnarray}\label{GME:1}
	C_{G M E}(\rho)=\min _{\left\{p_{i},\left|\varphi_{i}\right\rangle\right\}} \sum_{i} p_{i} C_{G M E}\left(\left|\varphi_{i}\right\rangle\right),
\end{eqnarray}
where the minimum is taken over all possible ensemble decompositions of
$\rho=\sum\limits_{i}p_{i}\left|\varphi_{i}\right\rangle\left\langle\varphi_{i}\right|$, $p_{i} \geqslant 0$ with $\quad\sum\limits_{i}p_{i}=1$.
We have the following theorem about the lower bound of GME concurrence, see proof in Appendix C.

\begin{theorem}\label{th:7}
For any tripartite state $\rho\in\mathbb{C}^{d} \otimes \mathbb{C}^{d} \otimes \mathbb{C}^{d}$, we have
\begin{eqnarray}\notag
		\begin{aligned}
			C_{G M E}(\rho)&\geq \frac{1}{\sqrt{d(d-1)}}{\bigg(\mathscr{Q}(\rho)\bigg.}\\&{\bigg.-\sqrt{\left(| \mu|^{2}+1\right)\left(|\nu|^{2}+1\right)}-\frac{2(d-1)}{3}\bigg)},
		\end{aligned}
\end{eqnarray}
where $\mathscr{Q}(\rho)$ is defined in (\ref{M:1}).
\end{theorem}

\begin{example}
Consider the following state $\rho_{x}\in\mathbb{C}^2 \otimes \mathbb{C}^2 \otimes \mathbb{C}^2$,
	$$
	\rho_{x}=\frac{x}{8} I_{8}+(1-x)\left|\phi_{G H Z}\right\rangle\left\langle\phi_{G H Z}\right|, 0 \leq x \leq 1,
	$$
where $\left|\phi_{G H Z}\right\rangle=\frac{1}{\sqrt{2}}(|000\rangle+|111\rangle)$, $I_{8}$ is the $8 \times 8$ identity matrix.
\end{example}

Set $\mu=(1,2)^T$ and $\nu=(2,1)^T$. From Theorem \ref{th:7} we get $C_{G M E}(\rho_{x}) \geq\frac{3\sqrt{2}\left(1-x\right)}{4}+\frac{\sqrt{11\left(x^2-2x+22\right)}}{4}
-\frac{17\sqrt{2}}{6}$. Fig.\ref{fig:5} shows that $\rho_{x}$ is genuine tripartite entangled for $0 \leq x \leq 0.193038$. When $\mu=\nu=(\sqrt{2},\sqrt{2})^T$, Theorem \ref{th:7} reduces to the Theorem $6$ in Ref.\cite{sun2024separability}, which detects the genuine tripartite entanglement of $\rho_{x}$ for $0 \leq x \leq 0.192912$. When $\mu=\nu= 1$, Theorem \ref{th:7} reduces to the Theorem $3$ in Ref. \cite{qi2024detection}, which detects the genuine tripartite entanglement of $\rho_{x}$ for $0 \leq x \textless 0.192758$. It can be seen that our Theorem \ref{th:7} detects better the genuine tripartite entanglement, lower bound of GME concurrence for $\rho_{x}$ in Fig. \ref{fig:5}.

\begin{figure}[h]
	\centering
	\includegraphics[width=90mm]{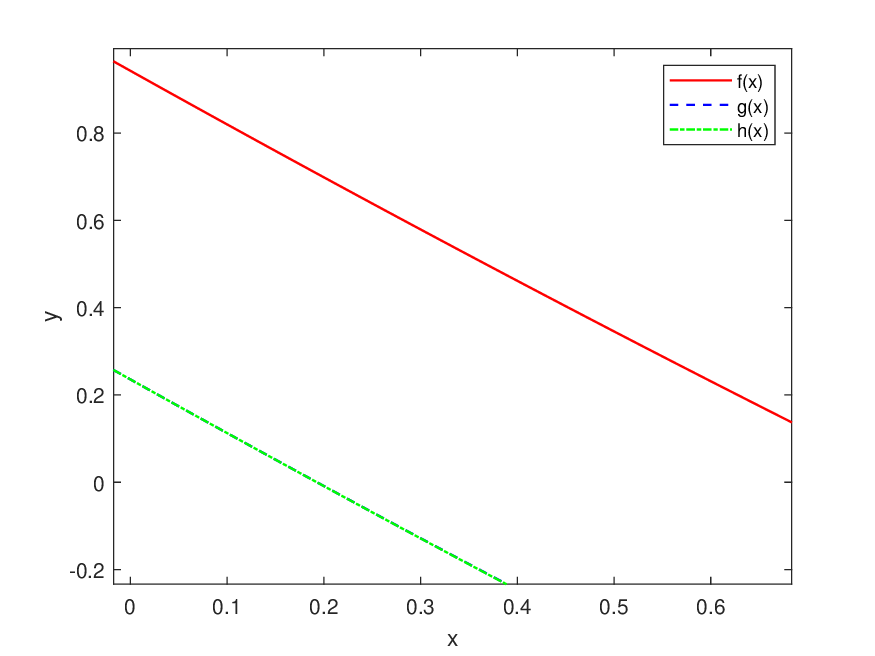}\\
\caption{Lower bound of GME concurrence for $\rho_x$,$f(x)$ is our result from Theorem \ref{th:7} (red solid line), $g(x)$ from Theorem $6$ in Ref.\cite{sun2024separability}  (dash blue line), $h(x)$ from Theorem $3$ in Ref.\cite{qi2024detection} (dash-dotted green line).}
\label{fig:5}
\end{figure}

\section{Conclusions and discussions}\label{S:4}

In this paper, we introduce a new set of separability criteria by constructing matrices derived from the realignment of density matrices. The proposed separability criteria can detect more entanglement than the previous separability criteria. Moreover, we provide criteria for detecting the genuine tripartite entanglement and present lower bounds for the concurrence and convex-roof extended negativity. The advantages of results are demonstrated through detailed examples. Furthermore, the methods presented here can be applied to derive lower bounds for certain multipartite entanglement measures. It is our hope that this work will provide insights and guidance for future research in the field.

\bigskip

\bigskip
\section*{APPENDIX}
\setcounter{equation}{0} \renewcommand%
\theequation{A\arabic{equation}}
\subsection{Proof of Lemma 1}\label{app:1}
\begin{proof}
(1) The vectorization of matrices has the following properties,
\begin{eqnarray}\label{eq:v2}
\operatorname{Vec}\left(k_{1} A+k_{2} B\right)=k_{1} \operatorname{Vec}(A)+k_{2}\operatorname{Vec}(B)
\end{eqnarray}
for any $A, B \in \mathbb{C}^{m \times n}$, $k_{i} \in \mathbb{R}$ $(i=1,2)$ and
\begin{eqnarray}\label{eq:v3}
\operatorname{Vec}(A B C)=\left(C^{T} \otimes A\right) \operatorname{Vec}(B)
\end{eqnarray}
for any $A \in \mathbb{C}^{m \times n}$, $B \in \mathbb{C}^{n \times p}$ and $C \in \mathbb{C}^{p \times q}$.

From (\ref{eq:v2}), it yields that for any $A, B \in \mathbb{C}^{mn \times mn}$, $k_{i} \in \mathbb{R}$ $(i=1,2)$,
\begin{eqnarray}\label{eq:r2}
\mathcal{R}\left(k_{1} A+k_{2} B\right)=k_{1} \mathcal{R}(A)+k_{2} \mathcal{R}(B).
\end{eqnarray}


From the definition of $\mathcal{Q}_{\mu,\nu}^{l}\left(\rho\right)$ in (\ref{eq:m1}), (\ref{eq:v2}) and (\ref{eq:r2}), we get
	\begin{eqnarray}\notag
			&& \mathcal{Q}_{\mu, \nu}\left(\sum_{i=1}^{n} k_{i} \rho_{i}\right)\\ &=&\left(\begin{array}{cc}
				\mu \nu^T & \mu \operatorname{Vec}\left(\operatorname{T r}_{A}\left(\sum\limits_{i=1}^{n} k_{i} \rho_{i}\right)\right)^{T}\\
				\operatorname{Vec}\left(\operatorname{T r}_{B}\left(\sum\limits_{i=1}^{n} k_{i} \rho_{i}\right)\right)\nu ^T & \mathcal{R}\left(\sum\limits_{i=1}^{n} k_{i} \rho_{i}\right)
			\end{array}\right)\notag \\
			& =&\left(\begin{array}{cc}
				\sum\limits_{i=1}^{n} k_{i} \mu \nu^T& \mu \operatorname{Vec}\left(\sum\limits_{i=1}^{n} k_{i}\operatorname{T r}_{A}\left(\rho_{i}\right)\right)^{T} \\
				\operatorname{Vec}\left(\sum\limits_{i=1}^{n} k_{i} \operatorname{T r}_{B}\left(\rho_{i}\right)\right)\nu ^T & \sum\limits_{i=1}^{n} k_{i} \mathcal{R}\left(\rho_{i}\right)
			\end{array}\right)\notag \\
			& =&\left(\begin{array}{cc}
				\sum\limits_{i=1}^{n} k_{i} \mu \nu^T& \mu \sum\limits_{i=1}^{n} k_{i} \operatorname{Vec}\left(\operatorname{T r}_{A}\left(\rho_{i}\right)\right)^{T} \\
				\sum\limits_{i=1}^{n} k_{i} \operatorname{Vec}\left(\operatorname{T r}_{B}\left(\rho_{i}\right)\right)\nu ^T & \sum\limits_{i=1}^{n} k_{i} \mathcal{R}\left(\rho_{i}\right)
			\end{array}\right)\notag \\
			& =&\sum\limits_{i=1}^{n} k_{i}\left(\begin{array}{cc}
				\mu \nu^T& \mu \operatorname{Vec}\left(\operatorname{T r}_{A}\left(\rho_{i}\right)\right)^{T} \\
				\operatorname{Vec}\left(\operatorname{T r}_{B}\left(\rho_{i}\right)\right)\nu ^T & \mathcal{R}\left(\rho_{i}\right)
			\end{array}\right)\notag \\
			& =&\sum\limits_{i=1}^{n} k_{i} \mathcal{Q}_{\mu, \nu}\left(\rho_{i}\right)\notag,
	\end{eqnarray}
which proves (1) of Lemma \ref{P:1}.
	
(2) A general state $\rho$ can be written as \cite{rudolph2005further} $\rho=\sum\limits_{i} \xi_{i} \otimes \eta_{i}$, where $\xi_{i} \in \mathbb{C}^{d_{A} \times d_{A}}$ and $\eta_{i} \in \mathbb{C}^{d_{B} \times d_{B}}$. Denote $\sigma=(U \otimes V) \rho\left(U \otimes V\right)^{\dagger}$. We have
	\begin{eqnarray}\label{eq:m5}
		\mathcal{Q}_{\mu, \nu}(\sigma)=\left(\begin{array}{cc}
			\mu \nu^T & \mu \operatorname{Vec}\left(\operatorname{T r}_{A}(\sigma)\right)^{T} \\
			\operatorname{Vec}\left(\operatorname{T r}_{B}(\sigma)\right)\nu^T & \mathcal{R}(\sigma)
		\end{array}\right),
	\end{eqnarray}
	where
\begin{eqnarray}\notag
		&& \mu \operatorname{Vec}\left(\operatorname{T r}_{A}(\sigma)\right)^{T}
		\\ \notag& =&\mu \operatorname{Vec}\left(\operatorname{T r}_{A}\left((U \otimes V) \sum_{i}\left(\xi_{i} \otimes \eta_{i}\right)\left(U \otimes V\right)^{\dagger}\right)\right)^{T} \\\notag
		& =&\mu \operatorname{Vec}\left(\sum_{i} \operatorname{T r}_{A}\left(U \xi_{i} U^{\dagger} \otimes V \eta_{i} V^{\dagger}\right)\right)^{T} \\\notag
		& =&\mu \operatorname{Vec}\left(\sum_{i} \operatorname{T r}\left(\xi_{i}\right) V \eta_{i} V^{\dagger}\right)^{T}.
\end{eqnarray}
From (\ref{eq:v2}) and (\ref{eq:v3}), it yields that
	\begin{eqnarray}\label{eq:m7}\notag
			&& \mu \operatorname{Vec}\left(\operatorname{T r}_{A}(\sigma)\right)^{T}\\\notag&=&\mu \sum_{i} \operatorname{T r}\left(\xi_{i}\right) \operatorname{Vec}\left(V \eta_{i} V^{\dagger}\right)^{T} \\\notag
			&
			=&\mu \sum_{i} \operatorname{T r}\left(\xi_{i}\right)\left(\left(\left(V^{\dagger}\right)^{T} \otimes V\right) \operatorname{Vec}\left(\eta_{i}\right)\right)^{T} \\\notag
			& =&\mu \sum_{i} \operatorname{T r}\left(\xi_{i}\right)\left((\overline{V} \otimes V) \operatorname{Vec}\left(\eta_{i}\right)\right)^{T} \\\notag
			& =&\mu \sum_{i} \operatorname{T r}\left(\xi_{i}\right) \operatorname{Vec}\left(\eta_{i}\right)^{T}(\overline{V} \otimes V)^{T} \\\notag
			& =&\mu \operatorname{Vec}\left(\sum_{i} \operatorname{T r}_{A}\left(\xi_{i} \otimes \eta_{i}\right)\right)^{T}(\overline{V} \otimes V)^{T} \\
			& =&\mu \operatorname{Vec}\left(\operatorname{T r}_{A}(\rho)\right)^{T}(\overline{V} \otimes V)^{T}.
	\end{eqnarray}
Similarly, we have
	\begin{eqnarray}\label{eq:m8}
		\operatorname{Vec}\left(\operatorname{T r}_{B}(\sigma)\right)\nu^T=(\overline{U} \otimes U) \operatorname{Vec}\left(\operatorname{T r}_{B}(\rho)\right)\nu ^T.
	\end{eqnarray}

Using (\ref{eq:v3}), (\ref{eq:r2}) and (\ref{eq:r4}) we obtain
	\begin{eqnarray}\label{eq:m9}\notag
			&&\mathcal{R}(\sigma)\\\notag
			& =&\sum_{i} \mathcal{R}\left((U \otimes V)\left(\xi_{i} \otimes \eta_{i}\right)\left(U \otimes V\right)^{\dagger}\right) \\\notag
			& =&\sum_{i} \mathcal{R}\left(U \xi_{i} U^{\dagger} \otimes V \eta_{i} V^{\dagger}\right) \\\notag
			& =&\sum_{i} \operatorname{Vec}\left(U \xi_{i} U^{\dagger}\right) \operatorname{Vec}\left(V \eta_{i} V^{\dagger}\right)^{T} \\\notag
			& =&\sum_{i}\left(\left(U^{\dagger}\right)^{T} \otimes U\right) \operatorname{Vec}\!\left(\xi_{i}\right)\!\left(\!\left(\left(V^{\dagger}\right)^{T} \otimes V\right) \operatorname{Vec}\!\left(\!\eta_{i}\!\right)\!\right)^{T} \\\notag
			& =&\sum_{i}(\overline{U} \otimes U) \operatorname{Vec}\left(\xi_{i}\right) \operatorname{Vec}\left(\eta_{i}\right)^{T}(\overline{V} \otimes V)^{T} \\\notag
			& =&\sum_{i}(\overline{U} \otimes U) \mathcal{R}\left(\xi_{i} \otimes \eta_{i}\right)(\overline{V} \otimes V)^{T} \\
			& =&(\overline{U} \otimes U) \mathcal{R}(\rho)(\overline{V} \otimes V)^{T}.
	\end{eqnarray}
Therefore, combing (\ref{eq:m5}), (\ref{eq:m7}), (\ref{eq:m8}) and (\ref{eq:m9}) we have
	\begin{eqnarray}\notag
			&&\mathcal{Q}_{\mu, \nu}(\sigma)\\&=&\left(\!\!\begin{array}{cc}
				\mu \nu^T& \mu \operatorname{Vec}\left(\operatorname{T r}_{A}(\rho)\right)^{T}(\overline{V} \otimes V)^{T} \\
				(\overline{U} \otimes U) \operatorname{Vec}\left(\operatorname{T r}_{B}(\rho)\right)\nu^T & (\overline{U} \otimes U) \mathcal{R}(\rho)(\overline{V} \otimes V)^{T}
			\end{array}\!\!\right)\notag\\\notag
			& =&\widetilde{U}\left(\begin{array}{cc}
				\mu \nu E_{l \times l} & \mu \operatorname{Vec}\left(\operatorname{T r}_{A}(\rho)\right)^{T} \\
				\nu \operatorname{Vec}\left(\operatorname{T r}_{B}(\rho)\right) & \mathcal{R}(\rho)
			\end{array}\right)\widetilde{V}\\\notag
			&=&\widetilde{U} \mathcal{Q}_{\mu, \nu}(\rho) \widetilde{V},
	\end{eqnarray}
where $\overline{U}=\left(U^{\dagger}\right)^{T}$,
$$
\widetilde{U}=\left(\begin{array}{cc}I_{l \times l} & 0 \\ 0 & \overline{U} \otimes U\end{array}\right),~~ \widetilde{V}=\left(\begin{array}{cc}I_{l \times l} & 0 \\ 0 & V \otimes \overline{V}\end{array}\right)^{\dagger}.
$$
Hence,
	\begin{eqnarray}\notag
			&&\left\|\mathcal{Q}_{\mu, \nu}\left((U \otimes V) \rho\left(U \otimes V\right)^{\dagger}\right)\right\|_{\mathrm{T r}}\\\notag&=&\left\|\widetilde{U} \mathcal{Q}_{\mu, \nu}(\rho) \widetilde{V}\right\|_{\mathrm{T r}}\\\notag&=&\left\|\mathcal{Q}_{\mu, \nu}(\rho)\right\|_{\mathrm{T r}},
	\end{eqnarray}
which proves (2) of Lemma \ref{P:1}.
\end{proof}

\subsection{Proof of Lemma 2}\label{app:2}
\begin{proof}
	$\left(1\right)$ Since
$|\varphi\rangle=\sum\limits_{i=0}^{d-1} \sqrt{\lambda_{i}}\left|i_{A} i_{B}\right\rangle$, one has
$$|\varphi\rangle\langle\varphi|=\sum\limits_{i=0}^{d-1} \sum\limits_{j=0}^{d-1} \sqrt{\lambda_{i} \lambda_{j}}\left|i_{A} i_{B}\right\rangle\left\langle j_{A} j_{B}\right|.$$
Hence,
	\begin{eqnarray}\notag
		\mathcal{Q}_{\!\mu, \nu}(|\varphi\rangle\langle\varphi|)\!=\left(\begin{array}{cc}
			\!\!\mu \nu^T & \mu \operatorname{Vec}\left(\operatorname{T r}_{A}(|\varphi\rangle\langle\varphi|)\right)^{T} \\
			\!\!\operatorname{Vec}\left(\operatorname{T r}_{B}(|\varphi\rangle\langle\varphi|)\right)\nu^T & \!\!\mathcal{R}(|\varphi\rangle\langle\varphi|)
		\end{array}\right).
	\end{eqnarray}

Let
	\begin{eqnarray}\notag
		\mathrm{Q}_{1}=\left(\begin{array}{cc}
			\mu \nu^T& \mu M \\
			 M^{T}\nu^T & \Omega_{1}
		\end{array}\right),~~ \mathrm{Q}_{2}=\left(\begin{array}{cc}
			0 & 0 \\
			0 & \Omega_{2}
		\end{array}\right),
	\end{eqnarray}
	where
	\begin{widetext}
		\begin{eqnarray}\notag
			M=(
\lambda_0,\underbrace{0,\cdots,0}_{d-1},0,\lambda_1,\underbrace{0,\cdots,0}_{d-1},\cdots,\lambda_{d-1}
)_{1 \times d^{2}},
		\end{eqnarray}
		\begin{eqnarray}\notag
			\Omega_{1}=\operatorname{diag}(\lambda_{0}, \underbrace{0, \ldots, 0}_{d-1}, 0, \lambda_{1}, \underbrace{0, \ldots, 0}_{d-2}, 0,0, \lambda_{2}, \underbrace{0, \ldots, 0}_{d-3}, \ldots, 0, \ldots, 0, \lambda_{d-1})\in \mathbb{C}^{d^{2} \times d^{2}},
		\end{eqnarray}
		\begin{eqnarray}\notag
			\begin{aligned}
				& \Omega_{2}=\operatorname{diag}\left(0, \sqrt{\lambda_{0} \lambda_{1}}, \ldots, \sqrt{\lambda_{0} \lambda_{d-1}}, \sqrt{\lambda_{1} \lambda_{0}}, 0, \sqrt{\lambda_{1} \lambda_{2}}, \ldots, \sqrt{\lambda_{1} \lambda_{d-1}}\right. \\
				& \left.\sqrt{\lambda_{2} \lambda_{0}}, \sqrt{\lambda_{2} \lambda_{1}}, 0, \sqrt{\lambda_{2} \lambda_{3}}, \ldots, \sqrt{\lambda_{2} \lambda_{d-1}}, \ldots, \sqrt{\lambda_{d-1} \lambda_{0}}, \ldots, \sqrt{\lambda_{d-1} \lambda_{d-2}}, 0\right)\in \mathbb{C}^{d^{2} \times d^{2}}.
			\end{aligned}
		\end{eqnarray}
	\end{widetext}
Using the definition of trace norm we get
	\begin{eqnarray}\notag
		\left\|\mathcal{Q}_{\mu, \nu}(|\varphi\rangle\langle\varphi|)\right\|_{\mathrm{T r}}=\left\|\mathrm{Q}_{1}\right\|_{\mathrm{T r}}+\left\|\mathrm{Q}_{2}\right\|_{\mathrm{T r}},
	\end{eqnarray}
	where $\left\|\mathrm{Q}_{2}\right\|_{\mathrm{T r}}=2 \sum\limits_{0 \leq i<j \leq d-1} \sqrt{\lambda_{i} \lambda_{j}}$. Since $\Omega_{1}=\mathcal{R}\left(\Omega_{1}\right)$ and $\Omega_{1}$ is a separable state, according to Theorem \ref{th:1} we get
	\begin{eqnarray}\notag
		\left\|\mathrm{Q}_{1}\right\|_{\mathrm{T r}}=\left\|\mathcal{Q}_{\mu, \nu}\left(\Omega_{1}\right)\right\|_{\mathrm{T r}} \leq \sqrt{\left(| \mu|^{2}+1\right)\left(|\nu|^{2}+1\right)}.
	\end{eqnarray}
	Therefore,
	\begin{eqnarray}\notag
			&& \left\|\mathcal{Q}_{\mu, \nu}(|\varphi\rangle\langle\varphi|)\right\|_{\mathrm{T r}} \\ \notag&\leq& \sqrt{\left(| \mu|^{2}+1\right)\left(|\nu|^{2}+1\right)}+2 \sum_{0 \leq i<j \leq d-1} \sqrt{\lambda_{i} \lambda_{j}}.
	\end{eqnarray}

$\left(2\right)$
It has been proved in \cite{qi2024detection} that for any $\lambda_{i}>0$, $i=0,1, \ldots, d-1$, such that $\lambda_{0}+\lambda_{1}+\cdots+\lambda_{d-1}=1$, one has
\begin{eqnarray}\notag
		2 \sum\limits_{0 \leq i<j \leq d-1} \sqrt{\lambda_{i} \lambda_{j}} \leq d-1.
\end{eqnarray}
Based on above relation and (1) of Lemma \ref{le:2}, we get
	\begin{eqnarray}\notag
			&&\left\|\mathcal{Q}_{\mu, \nu}(|\varphi\rangle\langle\varphi|)\right\|_{\mathrm{T r}} \\\notag&
			\leq &\sqrt{\left(|\mu|^{2}+1\right)\left(|\nu|^{2}+1\right)}+2 \sum_{0 \leq i<j \leq d-1} \sqrt{\lambda_{i} \lambda_{j}} \\\notag
			& \leq& \sqrt{\left(|\mu|^{2}+1\right)\left(| \nu|^{2}+1\right)}+(d-1).
	\end{eqnarray}
Therefore, we complete the proof of (2) in Lemma \ref{le:2}.
\end{proof}

\subsection{Proof of Theorem 6}\label{app:3}
\begin{proof}
	For a pure state $|\varphi\rangle$ in $\mathbb{C}^{d_{1}} \otimes \mathbb{C}^{d_{2}} \otimes \mathbb{C}^{d_{3}}$, using Theorem \ref{th:2} and (\ref{C:1}) we get
\begin{eqnarray}\label{B:1}
\begin{aligned}
			&\sqrt{1-\operatorname{T r}\left(\rho_{1}^{2}\right)} \geq \frac{1}{\sqrt{d(d-1)}}{\Big(\left\|\mathscr{Q}_{1|23}\left(|\varphi\rangle\langle\varphi|\right)\right\|\bigg.}\\&{\Big.\qquad\quad\qquad\qquad-\sqrt{\left(| \mu|^{2}+1\right)\left(|\nu|^{2}+1\right)}\Big)},
		\end{aligned}
\end{eqnarray}
\begin{eqnarray}\label{B:2}
		\begin{aligned}
			&\sqrt{1-\operatorname{T r}\left(\rho_{2}^{2}\right)} \geq \frac{1}{\sqrt{d(d-1)}}{\Big(\left\|\mathscr{Q}_{2|13}\left(|\varphi\rangle\langle\varphi|\right)\right\|\Big.}\\&{\Big.\qquad\quad\qquad\qquad-\sqrt{\left(\| \mu\|^{2}+1\right)\left(|\nu|^{2}+1\right)}\Big)}
		\end{aligned}
	\end{eqnarray}
and	
\begin{eqnarray}\label{B:3}
		\begin{aligned}
			&\sqrt{1-\operatorname{T r}\left(\rho_{3}^{2}\right)} \geq \frac{1}{\sqrt{d(d-1)}}{\Big(\left\|\mathscr{Q}_{3|12}\left(|\varphi\rangle\langle\varphi|\right)\right\|\Big.}\\&{\Big.\qquad\quad\qquad\qquad-\sqrt{\left(| \mu|^{2}+1\right)\left(|\nu|^{2}+1\right)}\Big)}.
		\end{aligned}
	\end{eqnarray}
Combining (\ref{B:1}), (\ref{B:2}) and (\ref{B:3}) we get
	\begin{eqnarray}\notag
		\begin{aligned}
			& \quad3 \sqrt{d(d-1)} \sqrt{1-\operatorname{T r}\left(\rho_{1}^{2}\right)}-3\mathscr{Q}(|\varphi\rangle\langle\varphi|) \\&\quad+\left(3 \sqrt{\left(| \mu|^{2}+1\right)\left(|\nu|^{2}+1\right)}+2(d-1)\right) \\
			& =3 \sqrt{d(d-1)} \sqrt{1-\operatorname{T r}\left(\rho_{1}^{2}\right)}-\left(\left\|\mathscr{Q}_{1|23}\left(|\varphi\rangle\langle\varphi|\right)\right\|_{\mathrm{T r}}\right.\\&\left.\quad+\left\|\mathscr{Q}_{2|13}\left(|\varphi\rangle\langle\varphi|\right)\right\|_{\mathrm{T r}}+\left\|\mathscr{Q}_{3|12}\left(|\varphi\rangle\langle\varphi|\right)\right\|_{\mathrm{T r}}\right) \\
			& \quad+\left(3 \sqrt{\left(| \mu|^{2}+1\right)\left(|\nu|^{2}+1\right)}+2(d-1)\right) \\
			& \geq 3 \sqrt{d(d-1)} \sqrt{1-\operatorname{T r}\left(\rho_{1}^{2}\right)}\!\!-\!\sqrt{d(d-1)}\left(\sqrt{1-\operatorname{T r}\left(\rho_{1}^{2}\right)}\right.\\&\left.\quad+\sqrt{1-\operatorname{T r}\left(\rho_{2}^{2}\right)}+\sqrt{1-\operatorname{T r}\left(\rho_{3}^{2}\right)}\right)+2(d-1) \\
			& =2 \sqrt{d(d-1)} \sqrt{1-\operatorname{T r}\left(\rho_{1}^{2}\right)}-\sqrt{d(d-1)}
\left(\sqrt{1-\operatorname{T r}\left(\rho_{2}^{2}\right)}
\right.\\&\left.\quad+\sqrt{1-\operatorname{T r}\left(\rho_{3}^{2}\right)}\right)+2(d-1).
		\end{aligned}
	\end{eqnarray}

Since $\sqrt{1-\operatorname{T r}\left(\rho_{1}^{2}\right)} \geq 0$ and $\sqrt{1-\operatorname{T r}\left(\rho_{i}^{2}\right)} \leq \sqrt{1-\frac{1}{d}}(i=2,3)$, we obtain
\begin{eqnarray}\notag
		\begin{aligned}
			& \quad 3 \sqrt{d(d-1)} \sqrt{1-\operatorname{T r}\left(\rho_{1}^{2}\right)}-3\mathscr{Q}(|\varphi\rangle\langle\varphi|)\\&\quad+\left(3 \sqrt{\left(|\mu|^{2}+1\right)\left(|\nu|^{2}+1\right)}+2(d-1)\right) \\
			& \geq2 \sqrt{d(d-1)} \sqrt{1-\operatorname{T r}\left(\rho_{1}^{2}\right)}\!-\!\!\sqrt{d(d-1)}\!\left(\sqrt{1-\operatorname{T r}\left(\rho_{2}^{2}\right)}\right.\\&\left.\quad+\sqrt{1-\operatorname{T r}\left(\rho_{3}^{2}\right)}\right)+2(d-1) \\
			& \geq 2(d-1)-2 \sqrt{d(d-1)} \sqrt{1-\frac{1}{d}}
\\
            & =0.
		\end{aligned}
	\end{eqnarray}
Therefore
	\begin{eqnarray}\notag
		\begin{aligned}
			&\sqrt{1-\operatorname{T r}\left(\rho_{1}^{2}\right)} \geq \frac{1}{\sqrt{d(d-1)}}{\bigg(\mathscr{Q}(|\varphi\rangle\langle\varphi|)\bigg.}\\&{\bigg.\qquad\quad\qquad\qquad-\sqrt{\left(| \mu|^{2}+1\right)\left(|\nu|^{2}+1\right)}-\frac{2(d-1)}{3}\bigg)}.
		\end{aligned}
	\end{eqnarray}
	Similarly, for $i=2,3$, we have
	\begin{eqnarray}\notag
		\begin{aligned}
			&\sqrt{1-\operatorname{T r}\left(\rho_{i}^{2}\right)} \geq \frac{1}{\sqrt{d(d-1)}}{\bigg(\mathscr{Q}(|\varphi\rangle\langle\varphi|)\bigg.}\\&{\bigg.\qquad\quad\qquad\qquad-\sqrt{\left(| \mu|^{2}+1\right)\left(|\nu|^{2}+1\right)}-\frac{2(d-1)}{3}\bigg)}.
		\end{aligned}
	\end{eqnarray}
	
Now let $\left\{p_{i},\left|\varphi_{i}\right\rangle\right\}$ be the optimal decomposition for $\rho$ such that $C_{G M E}(\rho)=\sum\limits_{i} p_{i} C_{G M E}\left(\left|\varphi_{i}\right\rangle\right)$. Then
	\begin{eqnarray}\notag
		\begin{aligned}
			C_{G M E}(\rho)&=\sum\limits_{i} p_{i} C_{G M E}\left(\left|\varphi_{i}\right\rangle\right) \\
			& \geq \frac{1}{\sqrt{d(d-1)}} \sum\limits_{i} p_{i}{\bigg(\mathscr{Q}(|\varphi_{i}\rangle\langle\varphi_{i}|)\bigg.}\\&{\bigg.\quad-\sqrt{\left(| \mu|^{2}+1\right)\left(|\nu|^{2}+1\right)}-\frac{2(d-1)}{3}\bigg)} \\
			& \geq \frac{1}{\sqrt{d(d-1)}}{\bigg(\mathscr{Q}(\rho)\bigg.}\\&{\bigg.\quad-\sqrt{\left(| \mu|^{2}+1\right)\left(|\nu|^{2}+1\right)}-\frac{2(d-1)}{3}\bigg)},
		\end{aligned}
	\end{eqnarray}
which completes proof of Theorem \ref{th:7}.
\end{proof}

\end{document}